\documentclass[12pt]{article}%
\IfFileExists{sariel_computer.sty}{\def\sarielComp{1}}{}
\ifx\sarielComp\undefined%
\newcommand{\SarielComp}[1]{}
\newcommand{\NotSarielComp}[1]{#1}%
\else
\newcommand{\SarielComp}[1]{#1}%
\newcommand{\NotSarielComp}[1]{}%
\fi
\newcommand{\IfPrinterVer}[2]{#2}%

\usepackage[cm]{fullpage}%
\usepackage{amsmath}%
\usepackage{amssymb}%
\usepackage{xcolor}%
\usepackage{mathtools}%

\SarielComp{\usepackage{sariel_colors}}%

\usepackage[amsmath,thmmarks]{ntheorem}%
\theoremseparator{.}%

\usepackage{titlesec}%
\titlelabel{\thetitle. }%
\usepackage{xcolor}%
\usepackage{mleftright}%
\usepackage{xspace}%
\usepackage{hyperref}%

\IfPrinterVer{%
   \usepackage{hyperref}%
}{%
   \usepackage{hyperref}%
   \hypersetup{%
      breaklinks,%
      ocgcolorlinks, colorlinks=true,%
      urlcolor=[rgb]{0.25,0.0,0.0},%
      linkcolor=[rgb]{0.5,0.0,0.0},%
      citecolor=[rgb]{0,0.2,0.445},%
      filecolor=[rgb]{0,0,0.4},
      anchorcolor=[rgb]={0.0,0.1,0.2}%
   }
}

\theoremseparator{.}%

\theoremstyle{plain}%
\newtheorem{theorem}{Theorem}[section]

\newtheorem{lemma}[theorem]{Lemma}

\newtheorem{observation}[theorem]{Observation}

\theoremstyle{plain}%
\theoremheaderfont{\sf} \theorembodyfont{\upshape}%
\newtheorem*{remark:unnumbered}[theorem]{Remark}%
\newtheorem{definition}[theorem]{Definition}

\newtheorem{problem}[theorem]{Problem}

\newcommand{\myqedsymbol}{\rule{2mm}{2mm}}

\theoremheaderfont{\em}%
\theorembodyfont{\upshape}%
\theoremstyle{nonumberplain}%
\theoremseparator{}%
\theoremsymbol{\myqedsymbol}%
\newtheorem{proof}{Proof:}%

\definecolor{blue25emph}{rgb}{0, 0, 11}
\providecommand{\emphic}[2]{%
   \textcolor{blue25emph}{%
      \textbf{\emph{#1}}}%
   \index{#2}}

\providecommand{\emphi}[1]{\emphic{#1}{#1}}

\definecolor{almostblack}{rgb}{0, 0, 0.3}
\providecommand{\emphw}[1]{{\textcolor{almostblack}{\emph{#1}}}}%

\newcommand{\atgen}{\symbol{'100}}
\newcommand{\SarielThanks}[1]{\thanks{Department of Computer Science;
      University of Illinois; 201 N. Goodwin Avenue; Urbana, IL,
      61801, USA; {\tt sariel\atgen{}illinois.edu}; {\tt
         \url{http://sarielhp.org/}.} #1}}

\newcommand{\HLink}[2]{\hyperref[#2]{#1~\ref*{#2}}}
\newcommand{\HLinkSuffix}[3]{\hyperref[#2]{#1\ref*{#2}{#3}}}

\newcommand{\thmlab}[1]{{\label{theo:#1}}}
\newcommand{\thmref}[1]{\HLink{Theorem}{theo:#1}}

\newcommand{\lemlab}[1]{\label{lemma:#1}}
\newcommand{\lemref}[1]{\HLink{Lemma}{lemma:#1}}%

\newcommand{\problab}[1]{\label{prob:#1}}%
\newcommand{\probref}[1]{\HLink{Problem}{prob:#1}}%

\providecommand{\eqlab}[1]{}%
\renewcommand{\eqlab}[1]{\label{equation:#1}}

\newcommand{\remove}[1]{}%

\newcommand{\brc}[1]{\left\{ {#1} \right\}}
\newcommand{\cardin}[1]{\left| {#1} \right|}%

\usepackage[inline]{enumitem}

\newlist{compactenumA}{enumerate}{5}%
\setlist[compactenumA]{topsep=0pt,itemsep=-1ex,partopsep=1ex,parsep=1ex,%
   label=(\Alph*)}%

\newlist{compactenuma}{enumerate}{5}%
\setlist[compactenuma]{topsep=0pt,itemsep=-1ex,partopsep=1ex,parsep=1ex,%
   label=(\alph*)}%

\newlist{compactenumI}{enumerate}{5}%
\setlist[compactenumI]{topsep=0pt,itemsep=-1ex,partopsep=1ex,parsep=1ex,%
   label=(\Roman*)}%

\newlist{compactenumi}{enumerate}{5}%
\setlist[compactenumi]{topsep=0pt,itemsep=-1ex,partopsep=1ex,parsep=1ex,%
   label=(\roman*)}%

\newlist{compactitem}{itemize}{5}%
\setlist[compactitem]{topsep=0pt,itemsep=-1ex,partopsep=1ex,parsep=1ex,%
   label=\bullet}%

\DeclareMathOperator{\avg}{avg}

\newcommand{\MetricSpace}{\mathcal{X}}
\newcommand{\MetricDist}{d_{\MetricSpace{}}}

\newcommand{\hide}[1]{}

\newcommand{\SaveContent}[2]{%
   \expandafter\newcommand{#1}{#2}%
}

\newcommand{\RestatementOf}[2]{
   \noindent%
   \textbf{Restatement of #1.}
   {\em #2{}}%
}

\newcommand{\EliotThanks}[1]{%
   \thanks{%
      Department of Computer Science; University of Illinois; 201
      N. Goodwin Avenue; Urbana, IL, 61801, USA; {\tt
         erobson2\atgen{}illinois.edu}; {\tt
         \url{https://eliotwrobson.github.io/}.} #1}}\providecommand{\etal}{et~al.\xspace}
         
\renewcommand{\etal}{et~al.\xspace}

\title{Improving the average dilation of a metric graph by adding edges}

\author{
Sariel Har-Peled
\SarielThanks{}
\and
Eliot W. Robson%
\EliotThanks{}
}
\date{\today}

\begin{document}
\maketitle

\begin{abstract}
    For a graph \(G\) spanning a metric space, the \emphi{dilation} of
    a pair of points is the ratio of their distance in the shortest
    path graph metric to their distance in the metric space. Given a
    graph \(G\) and a budget \(k\), a classic problem is to augment
    \(G\) with \(k\) additional edges to reduce the maximum dilation
    \cite{fgg-fbsgn-05, DBLP:conf/soda/GudmundssonW21}.

    In this note, we consider a variant of this problem where the
    goal is to reduce the \emphw{average} dilation for pairs of points in
    $G$. We provide an $O(k)$ approximation algorithm for this
    problem, matching the approximation ratio given by
    \cite{DBLP:conf/soda/GudmundssonW21} for the maximum dilation
    variant.
\end{abstract}

\section{Introduction}

We consider a weighted, undirected graph $G = (V,E)$ and a metric
space $\MetricSpace = (X, \MetricDist)$ such that $V = X$, so both
$G$ and $\MetricSpace$ are defined over the same ground set.
For two points
$u,v \in X$, $\MetricDist (u,v)$ is the distance between $u$ and $v$ as
measured in the metric space $\MetricSpace$, and $d_G(u,v)$ is the distance
between $u$ and $v$ in the shortest path metric of $G$ (i.e.\ only traversing edges of $G$). Unless otherwise stated, we will assume that
$G$ is connected. For any edge $e = (u,v)\in E$, we require that
$w(e) = \MetricDist (u,v)$ (the weight of edges in the graph is
taken from the underlying metric space $\MetricSpace$).

\subsection{Basic definitions}

\begin{definition}
    For a pair of vertices $u,v \in V$ in a graph $G$ over an
    underlying metric space $\MetricSpace$, the \emphi{dilation} of
    the pair is
    \begin{equation*}
        s_G (u,v) = \frac{d_G (u,v)}{\MetricDist
           (u,v)}.
    \end{equation*}
    % We simply write $s(u,v)$ when $G$ is clear from
    % context.
\end{definition}

\begin{definition}
    The \emphi{average dilation} of a graph $G$ with $n$ vertices is
    \[
        s_{\avg} (G) \coloneqq \frac{1}{\binom{n}{2}} \sum_{u,v \in V}
        s_G(u,v).
    \]
\end{definition}

In this note, we study the following problem.

\begin{problem}
    \problab{min_avg_stretch} Given a positive integer $k$ and a graph
    $G = (V,E)$ over a metric space $\MetricSpace$, compute a set $F$
    of $k$ edges realizing
    \[
        \min_{\substack{\cardin{F} = k \\ F \subseteq V \times V}}
        s_{\avg} (G \cup F) = \min_{\substack{\cardin{F} = k \\ F
              \subseteq V \times V}} \frac{1}{\binom{n}{2}} \sum_{u,v
           \in V} s_{G \cup F} (u,v).
    \]
    %We use $\mathsf{AvgStretch}(F) = s_{\avg} (G \cup F)$.
\end{problem}

\subsection{Additional definitions}

We state some definitions specific to our main result.

\begin{definition}
    Let $F$ be a set of edges used to augment some graph $G = (V,E)$,
    and let $V_F \subset V$ be the set of endpoints of edges of
    $F$. Define the function
    $\delta_{F} : V \times V \rightarrow (V_F \times V_F) \cup
    \emptyset$ such that, for vertices $a,b \in V$, $\delta_{F} (a,b)$
    returns a tuple $(v_a, v_b)$ where $v_a$ is the first visited
    endpoint of the first edge of $F$ in the shortest path between $a$
    and $b$ in $G \cup F$, and $v_b$ is the last visited endpoint of
    the last edge of $F$ in the same shortest path. If no edges from
    $F$ are visited in the shortest path, then
    $\delta_{F} (a,b) = \emptyset$. The quantity $\delta_{F} (a,b)$ is
    the \emphi{signature} of $(a,b)$ in $F$.
\end{definition}

The function $\delta_{F}$ has $O( \cardin{F}^2) =O(k^2)$ distinct
values it might return.

\begin{definition}
    With respect to the average dilation function $s_{\avg}$ and graph
    $G$, the \emphi{benefit} of a set $F$ is
    $B(F) = s_{\avg} (G) - s_{\avg}(G \cup F) = \sum_{u,v \in V} b_F
    (u,v)$, where $b_F (u,v) = s_{G} (u,v) - s_{G \cup F}
    (u,v)$. Similarly, the \emphi{restricted benefit} for $a,b \in V$
    is the quantity
    $B|_{(a,b)} (F) = \sum_{u,v \in V: \, \delta_{F} (u,v) = (a,b) }
    b_F (u,v)$.
\end{definition}

The restricted benefit is the contribution from a single signature to the total benefit in the optimal solution. From this, the following
observation is immediate.

\begin{observation}
    For some graph $G = (V,E)$, we let $P$ be the set of all pairs of
    distinct vertices in $V(F)$. Then, we have
    % for any augmenting edge set $F$,
    $B(F) = \sum_{(a,b) \in P} B|_{(a,b)} (F)$.
\end{observation}

\subsection{Our contribution}

We analyze the greedy algorithm for placing $k$ additional edges. Assume the algorithm already computed a set of $i-1$ shortcut edges $F_{i-1}$.
In the $i$th iteration, the algorithm picks the edge $e_i$ that
maximizes $B(F_{i-1} \cup \brc{e_i})$. The algorithm then sets
$F_i = F_{i-1} \cup \brc{e_i}$ (here $F_0 = \emptyset$).

Our main result is that the output of this algorithm is an $O(k)$ 
approximation in the benefit for \probref{min_avg_stretch}.

\SaveContent{\MainThm}{
After $k$ steps, we have
    \begin{math}
        B(F_k) \geq B(F^*) / 8k. 
    \end{math}
    After $4k^2$ steps, we have
    \begin{math}
        B(F_{4k^2}) \geq B(F^*) / 2. 
    \end{math}
}

\begin{theorem}
    \thmlab{main}%
    \MainThm{}
\end{theorem}

\section{Related work}

The maximum dilation variant of \probref{min_avg_stretch} was
introduced by Farshi \etal \cite{fgg-fbsgn-05} and studied for
$k = 1$.  Extending these results for the case of $k > 1$ was posed as
an open problem in \cite[Open Problem 9]{ns-gsn-07}. This was answered
by Gudmundsson \etal \cite{DBLP:conf/soda/GudmundssonW21}, who gave an
$O(k)$ approximation algorithm.

%%%%%%%%%%%%%%%%%%%%%%%%%%%%%%%%%%%%%%%%%%%%%%%%%%%%%%%%%%%%%%% 5
\section{Approximation ratio of greedy algorithm}

\begin{lemma}
    \lemlab{key}
    Let $S$ be any set of shortcut edges, and $F^*$ be the optimal set
    of $k$ shortcuts. Then, either $B(S) \geq B(F^*)/2$, or there
    exists an edge $e$, such that $B(S + e ) \geq B(S) +B(F^*)/8k^2$.
\end{lemma}
\begin{proof}
    Assume $B(S) < B(F^*)/2$.  Observe that
    $B(S \cup F^*) \geq B(F^*)$. The are $k = |F^*|$ edges, and they
    have $2k$ endpoints, and thus overall there are
    $2k(2k-1)+1 \leq 4k^2$ signatures. One of the signatures, in
    relation to $S$, must have benefit $\geq (B(S \cup F^*) -B(S)) / (4k^2)$,
    as the total benefits of all the signatures exceeds $B(S \cup
    F^*)$. Thus, for the edge $e$ maximizing $B(S \cup e)$, we have
    \begin{equation*}
        B(S\cup e)
        \geq
        B(S) + \frac{B(S\cup F^*) - B(S)}{ 4k^2}
        \geq%
        B(S) + \frac{B(F^*) - B(F^*)/2}{ 4k^2}
        \geq%
        B(S) + \frac{B(F^*) }{8k^2}.
    \end{equation*}
\end{proof}

\RestatementOf{\thmref{main}}{\MainThm}

\begin{proof}
    After $k$ steps of the greedy algorithm, we have by \lemref{key} that 
    \begin{equation*}
        B(F_k) \geq B(F_{k-1}) +
        \frac{B(F^*) }{8k^2}
        \geq \cdots
        \geq
        B(F_0) + k \frac{B(F^*)}{8k^2}
        =
        0 + \frac{B(F^*)}{8k}.
    \end{equation*}
    The second part follows by similar argumentation.
\end{proof}

\bibliographystyle{alpha}%
\bibliography{min_stretch}

\begin{thebibliography}{FGG05}

\bibitem[FGG05]{fgg-fbsgn-05}
Mohammad Farshi, Panos Giannopoulos, and Joachim Gudmundsson.
\newblock Finding the best shortcut in a geometric network.
\newblock In Joseph S.~B. Mitchell and G{\"{u}}nter Rote, editors, {\em Proceedings of the 21st {ACM} Symposium on Computational Geometry, Pisa, Italy, June 6-8, 2005}, pages 327--335. {ACM}, 2005.

\bibitem[GW21]{DBLP:conf/soda/GudmundssonW21}
Joachim Gudmundsson and Sampson Wong.
\newblock Improving the dilation of a metric graph by adding edges.
\newblock In D{\'{a}}niel Marx, editor, {\em Proceedings of the 2021 {ACM-SIAM} Symposium on Discrete Algorithms, {SODA} 2021, Virtual Conference, January 10 - 13, 2021}, pages 675--683. {SIAM}, 2021.

\bibitem[NS07]{ns-gsn-07}
Giri Narasimhan and Michiel H.~M. Smid.
\newblock {\em Geometric spanner networks}.
\newblock Cambridge University Press, 2007.

\end{thebibliography}

\end{document}